\documentclass[runningheads]{llncs}
\usepackage[T1]{fontenc}
\usepackage[title]{appendix}
\usepackage{enumitem} 
\usepackage{amsmath,amsfonts,amssymb}
\usepackage{environ}
\usepackage{float}

\newcommand{\dtw}{\mathrm{DTW}}

\newcommand{\dd}{\mathinner{.\,.\allowbreak}}
\newcommand{\RLE}{\mathrm{RLE}}
\newcommand{\Oh}{\mathcal{O}}
\newcommand{\eps}{\varepsilon}

\newcommand{\repeattheorem}[1]{%
  \begingroup
  \renewcommand{\thetheorem}{\ref{#1}}%
  \expandafter\expandafter\expandafter\theorem
  \csname reptheorem@#1\endcsname
  \endtheorem
  \endgroup
}

\NewEnviron{reptheorem}[1]{%
  \global\expandafter\xdef\csname reptheorem@#1\endcsname{%
    \unexpanded\expandafter{\BODY}%
  }%
  \expandafter\theorem\BODY\unskip\label{#1}\endtheorem
}

\newcommand{\repeatlemma}[1]{%
  \begingroup
  \renewcommand{\thelemma}{\ref{#1}}%
  \expandafter\expandafter\expandafter\lemma
  \csname replemma@#1\endcsname
  \endlemma
  \endgroup
}

\NewEnviron{replemma}[1]{%
  \global\expandafter\xdef\csname replemma@#1\endcsname{%
    \unexpanded\expandafter{\BODY}%
  }%
  \expandafter\lemma\BODY\unskip\label{#1}\endlemma
}

\newtheorem{fact}[theorem]{Fact}
\newtheorem{observation}[theorem]{Observation}

\usepackage{subcaption,tikz}
\usetikzlibrary{calc} 
\usetikzlibrary{decorations.pathreplacing} 

\usepackage{todonotes}

\begin{document}
\title{Pattern matching under $\dtw$ distance\thanks{This work was partially funded by the grants ANR-20-CE48-0001, ANR-19-CE45-0008 SeqDigger and ANR-19-CE48-0016 from the French National Research Agency.}}
%
%
\author{Garance Gourdel\inst{1,2} \and
Anne Driemel\inst{3} \and
Pierre Peterlongo \inst{2} \and
Tatiana Starikovskaya \inst{1}}
\authorrunning{G. Gourdel et al.}
%
\institute{
DIENS, \'{E}cole normale sup\'{e}rieure de Paris, PSL Research University,  France
\email{\{garance.gourdel,tat.starikovskaya\}@gmail.com} \and
IRISA Inria Rennes, France \email{pierre.peterlongo@inria.fr} \and
Hausdorff Center for Mathematics, University of Bonn, Germany
\email{driemel@cs.uni-bonn.de}}
\maketitle              
\begin{abstract}
In this work, we consider the problem of pattern matching under the dynamic time warping ($\dtw$) distance motivated by potential applications in the analysis of biological data produced by the third generation sequencing. To measure the $\dtw$ distance between two strings, one must ``warp'' them, that is, double some letters in the strings to obtain two equal-lengths strings, and then sum the distances between the letters in the corresponding positions. When the distances between letters are integers, we show that for a pattern $P$ with $m$ runs and a text~$T$ with $n$ runs:
\begin{enumerate}
\item There is an $\Oh(m+n)$-time algorithm that computes all locations where the $\dtw$ distance from $P$ to $T$ is at most $1$;
\item There is an $\Oh(kmn)$-time algorithm that computes all locations where the $\dtw$ distance from $P$ to $T$ is at most $k$.
\end{enumerate}
As a corollary of the second result, we also derive an approximation algorithm for general metrics on the alphabet. 

\end{abstract}
%
%
%
\section{Introduction}
Introduced more than forty years ago~\cite{sakoe1978dynamic}, the dynamic time warping ($\dtw$) distance has become an essential tool in the time series analysis and its applications due to its ability to preserve the signal despite speed variation in compared sequences. To measure the $\dtw$ distance between two discrete temporal sequences, one must ``warp'' them, that is, replace some data items in the sequences with multiple copies of themselves to obtain two equal-lengths sequences, and then sum the distances between the data items in the corresponding positions. 

The $\dtw$ distance has been extensively studied for parametrised curves~--- sequences where the data items are points in a multidimensional space~---  specifically, in the context of locality sensitive hashing and nearest neighbour search~\cite{LSH,ANN}. In this work, we focus on a somewhat simpler, but surprisingly much less studied setting when the data items are elements of a finite set, the alphabet. Following traditions, we call such sequences \emph{strings}. 

The classical textbook dynamic programming algorithm computes the $\dtw$ distance between two $N$-length strings in $\Oh(N^2)$ time and space. Unfortunately, unless the Strong Exponential Time Hypothesis is false, there is no algorithm with strongly subquadratical time even for ternary alphabets~\cite{DBLP:conf/focs/AbboudBW15,DBLP:conf/focs/BringmannK15,DBLP:conf/icalp/Kuszmaul19}. On the other hand, very recently Gold and Sharir~\cite{DBLP:journals/talg/GoldS18} showed the first weakly subquadratic time algorithm (to be more precise, the time complexity of the algorithm is~$\Oh(N^2 \log \log  \log N / \log \log N)$). Kuszmaul~\cite{DBLP:conf/icalp/Kuszmaul19} gave a $\Oh(kN)$-time algorithm that computes the value of the distance between the strings if it is bounded by $k$, assuming that the distance between any two distinct letters of the alphabet is at least one, and used it to derive a subquadratic-time approximation algorithm for the general case. Finally, it is known that binary strings admit much faster algorithms: Abboud, Backurs, and Vassilevska Williams~\cite{DBLP:conf/focs/AbboudBW15} showed an $O(N^{1.87})$-time algorithm followed by a linear-time algorithm by Kuszmaul~\cite{DBLP:journals/corr/abs-2101-01108}. 

The problem of computing the $\dtw$ distance has also been studied in the sparse and run-length compressed settings, as well as in the low distance regime.  
In the sparse setting, we assume that most letters of the string are zeros. Hwang and Gelfand~\cite{hwang2017sparse} gave an $\Oh((s+t) N)$-time algorithm, where $s$ and $t$ denote the number of non-zero letters in each of the two strings. On sparse binary strings, the distance can be computed in $\Oh(s+t)$ time~\cite{DBLP:conf/mldm/HwangG19,mueen2016awarp}. Froese et al.~\cite{DBLP:journals/corr/abs-1903-03003} suggested an algorithm with running time $\Oh(mN+nM)$, where $M,N$ are the length of the strings, and $m, n$ are the sizes of their run length encodings. If $n \in \Oh(\sqrt{N})$ and $m \in \Oh(\sqrt{M})$, their algorithm runs in time $\Oh(nm \cdot (n+m))$. For binary strings, the $\dtw$ distance can be computed in $\Oh(nm)$ time~\cite{DBLP:conf/pkdd/DupontM15a}. 

Nishi et al.~\cite{DBLP:conf/spire/NishiNIBT20} considered the question of computing the $\dtw$ distance in the dynamic setting when the stings can be edited, and Sakai and Inenaga~\cite{DBLP:conf/isaac/SakaiI20} showed a reduction from the problem of computing the $\dtw$ distance to the problem of computing the longest increasing subsequence, which allowed them to give polynomial-time algorithms for a series of $\dtw$-related problems. 

In this work, we focus on the pattern matching variant of the problem: Given a pattern $P$ and a text $T$, one must output the smallest $\dtw$ distance between~$P$ and a suffix of $T[1 \dd r]$ for every position $r$ of the text. 

Our interest to this problem sparks from its potential applications in Third Generation Sequencing (TGS) data comparisons. TGS has changed the genomic landscape as it allows to sequence reads of few dozens of thousand of letters where previous sequencing techniques were limited to few hundred letters~\cite{amarasinghe2020opportunities}. However, TGS suffers from a high error rate (from $\approx$ 1 to 10\% depending on the used techniques) mainly due to the fact that the DNA sequences are read and thus sequenced at an uneven speed. The uneven sequencing speed has a major impact in the sequencing quality of DNA regions composed of two or more equal consecutive letters. Those regions, called \emph{homopolymers}, are hardly correctly sequenced as, due to the uneven sequencing speed, their size cannot be precisely determined~\cite{huang_homopolish_2021}. In particular, a common post-sequencing task  consists in aligning the obtained reads to a reference genome. This enables for instance to predict alternative splicing and gene expression~\cite{gonzalez2016introduction} or to detect structural variations~\cite{mahmoud2019structural}. All known aligners use the edit distance, most likely, due to the availability of software tools for the latter (see~\cite{10.1093/bioinformatics/bty191} and references therein). However, we find that the nature of TGS errors is much better described by the $\dtw$ distance, which we confirm experimentally in Section~\ref{sec:experiments}.

\paragraph{Our contribution.}  As a baseline, we show that the problem of pattern matching under the $\dtw$ distance can be solved using dynamic programming in time $\Oh(MN)$, where $M$ is the length of the pattern and $N$ of the text (Lemma~\ref{lm:recursion}). 

We then proceed to show more efficient algorithms for the low-distance regime on run-length compressible data, which is arguably the most interesting setting for the TGS data processing. Formally, in the \emph{$k$-$\dtw$ problem} we are given an integer $k > 0$, a pattern $P$ and a text $T$, and must find all positions $r$ of the text such that the smallest $\dtw$ distance between the pattern $P$ and a suffix of $T[1 \dd r]$ does not exceed~$k$. One might hope that the $\dtw$ distance is close enough to the edit distance and thus is amenable to the techniques developed for the latter, such as~\cite{LMS98,LV97}. In Appendix~\ref{sec:k=1}, we show that this is indeed the case for $k = 1$:

\begin{replemma}{1DTW}
\label{lm:1-DTW}
Given run-length encodings of a pattern $P$ and of a text $T$ over an alphabet $\Sigma$ and a distance $d: \Sigma \times \Sigma \rightarrow \mathbb{Z}^+$, the $1$-$\dtw$ problem can be solved in $\Oh(m+n)$ time, where $m$ is the number of runs in $P$ and $n$ is the number of runs in $T$. The output is given in a compressed form, with a possibility to retrieve each position in constant time.
\end{replemma}

Unfortunately, extending the approach of~\cite{LMS98,LV97} to higher values of $k$ seems to be impossible as it is heavily based on the fact that in the edit distance dynamic programming matrix the distances are non-decreasing on every diagonal, which is not the case for the $\dtw$ distance (see Fig.~\ref{fig:decreasing}). 

In Section~\ref{sec:block} we develop a different approach. Interestingly, we show that the value of any cell of the bottom row and the right column of a block of the dynamic programming table (i.e. a subtable formed by a run in the pattern and a run in the text) can be computed in constant time given a constant-time oracle access to the left column and the top row. Combining this with a compact representation of the $k$-bounded values, we obtain the following result:

\begin{theorem}\label{th:block}
Given run-length encodings of a pattern $P$ and of a text $T$ over an alphabet $\Sigma$ and a distance $d: \Sigma \times \Sigma \rightarrow \mathbb{Z}^+$, the $k$-$\dtw$ problem can be solved in $\Oh(k mn)$ time, where $m$ is the number of runs in $P$ and $n$ is the number of runs in $T$. The output is given in a compressed form, with a possibility to retrieve each position in constant time.
\end{theorem}

We note that while our algorithm can be significantly faster than the baseline one, its worst-case time complexity is cubic. We leave it as an open question whether there exists an $\Oh(k \cdot (m+n))$-time algorithm.  Finally, in Section~\ref{sec:approx} we use Theorem~\ref{th:block} to derive an approximation algorithm for the general variant of pattern matching under the $\dtw$ distance. 

\begin{figure}[H]

\begin{center}
\footnotesize
\resizebox{0.8\textwidth}{!}{
\begin{tabular}{|cc||cc|cccc|c|cc|c|cccc|cc|c|c|c|c|c|}
\hline
 &   & G & G & T & T & T & T & C & T & T & A & T & T & T & T & G & G & T & G & A & T & A \\
 & 0 & 0 & 0 & 0 & \textcolor{red}{0} & 0 & 0 & 0 & 0 & 0 & 0 & 0 & 0 & 0 & 0 & 0 & 0 & 0 & 0 & 0 & 0 & 0 \\
\hline
A  & $\infty$  & 1 & 1 & 1 & 1 & \textcolor{red}{1} & 1 & 1 & 1 & 1 &{ 0 } & 1 & 1 & 1 & 1 & 1 & 1 & 1 & 1 &{ 0 } & 1 &{ 0 }\\
A  & $\infty$  & 2 & 2 & 2 & 2 & 2 & \textcolor{red}{2} & 2 & 2 & 2 &{ 0 } & 1 & 2 & 2 & 2 & 2 & 2 & 2 & 2 &{ 0 } & 1 &{ 0 }\\
\hline
T  & $\infty$  & 3 & 3 &{ 2 } &{ 2 } &{ 2 } &{ 2 } & \textcolor{red}{3} &{{ 2 }} &{{ 2 }} & 1 &{ 0 } &{ 0 } &{ 0 } &{ 0 } & 1 & 2 &{ 2 } & 3 & 1 &{ 0 } & 1\\
T  & $\infty$  & 4 & 4 &{ 2 } &{ 2 } &{ 2 } &{ 2 } & 3 & \textcolor{red}{{ 2 }} &{{ 2 }} & 2 &{ 0 } &{ 0 } &{ 0 } &{ 0 } & 1 & 2 &{ 2 } & 3 & 2 &{ 0 } & 1\\
\hline
A  & $\infty$  & 5 & 5 & 3 & 3 & 3 & 3 & 3 & 3 & \textcolor{red}{3} &{{ 2 }} & 1 & 1 & 1 & 1 & 1 & 2 & 3 & 3 &{ 2 } & 1 &{ 0 }\\
\hline
T  & $\infty$  & 6 & 6 &{ 3 } &{ 3 } &{ 3 } &{ 3 } & 4 &{ 3 } &{ 3 } & \textcolor{red}{3} &{{ 1 }} &{{ 1 }} &{{ 1 }} &{{ 1 }} & 2 & 2 &{ 2 } & 3 & 3 &{ 1 } & 1\\
\hline

\end{tabular} }
\end{center}

\caption{Consider $P = AATTAT$ and $T=GGTTTTCTTATTTTGGTGATA$. A cell $(i,j)$ contains the smallest $\dtw$ distance between $P[1\dd i]$ and $T[1\dd j]$, where the distance between two letters equals one if they are distinct and zero otherwise. A non-monotone diagonal of the table is shown in red.}  
\label{fig:decreasing}
\end{figure}


\section{Preliminaries}\label{sec:prelim}
We assume a polynomial-size alphabet $\Sigma$ with $\sigma$ \emph{letters}. A \emph{string} $X$ is a sequence of letters. If the sequence has length zero, it is called the \emph{empty string}.  Otherwise, we assume that the letters in $X$ are numbered from $1$ to $n =: |X|$ and denote the $i$-th letter by $X[i]$. We define $X[i \dd j]$ to be equal to $X[i] \dots X[j]$ which we call a \emph{substring} of $X$ if $i \le j$ and to the empty string otherwise. If $j = n$, we call a substring $X[i \dd j]$ \emph{a suffix} of $X$.

\begin{definition}[Run, Run-length encoding]
A run of a string $X$ is a maximal substring $X[i \dd j]$ such that $X[i] = X[i+1] = \ldots = X[j]$. The run-length encoding of a string $X$, $\RLE(X)$ is a sequence obtained from $X$ by replacing each run with a tuple consisting of the letter forming the run and the length of the run. For example, $\RLE(aabbbc) = (a,2) (b,3)(c,1)$.
\end{definition}

Let $d: \Sigma \times \Sigma \rightarrow \mathbb{R}^+$ be a distance function such that for any letters $a,b \in \Sigma$, $a\neq b$, we have $d(a,a) = 0$ and $d(a,b) > 0$. The dynamic time warping distance $\dtw_d(X,Y)$ between strings $X,Y \in \Sigma^\ast$ is defined as follows. If both strings are empty, $\dtw_d(X,Y) = 0$. If one of the strings is empty, and the other is not, then  $\dtw_d(X,Y) = \infty$. Otherwise, let $X = X[1] X[2] \ldots X[r]$ and $Y = Y[1] Y[2] \ldots Y[q]$.  Consider an $r \times q$ grid graph such that each vertex $(i,j)$ has (at most) three outgoing edges: one going to $(i+1,j)$ (if it exists), one to $(i+1,j+1)$ (if it exists), and one to $(i,j+1)$ (if it exists). A path $\pi$ in the graph starting at $(1,1)$ and ending at $(r,q)$ is called a \emph{warping path}, and its \emph{cost} is defined to be $\sum_{(i,j) \in \pi} d(X[i],Y[j])$. Finally, $\dtw_d(X,Y)$ is defined to be the minimum cost of a warping path for $X,Y$. Below we omit $d$ if it is clear from the context. 

Let $M = |P|$, $N = |T|$, and $D$ be an $(M+1) \times (N+1)$ table where the rows are indexed from $0$ to $M$, and the columns from $0$ to $N$ such that:
\begin{enumerate}
\item For all $j \in [0, N]$, $D[0,j] = 0$;
\item For all $i \in [1, M]$, $D[i,0] = +\infty$;
\item For all $i \in [1, M]$ and $j \in [1, N]$, $D[i,j]$ equals the smallest $\dtw$ distance between $P[1\dd i]$ and a suffix of $T[1 \dd j]$. 
\end{enumerate}
(See Fig.~\ref{fig:decreasing}.) To solve the pattern matching problem under the $\dtw$ distance, it suffices to compute the last row of the table $D$. 

\begin{replemma}{basicrecursion}\label{lm:recursion}
The table $D$ can be computed in $\Oh(MN)$ time via a dynamic programming algorithm, using the following recursion for all $1 \le i \le M, 1 \le j \le N$:
\begin{align*}
D[i,j] = 
\min\{D[i-1,j-1],D[i-1,j], D[i,j-1]\}+ d(P[i], T[j])
\end{align*}
\end{replemma}

\vspace{-20pt}

In the subsequent sections, we develop more efficient solutions for the low-distance regime on run-length compressible data. We will be processing the table $D$ by blocks, defined as follows: A subtable $D[i_p \dd j_p, i_t \dd j_t]$ is called a \emph{block} if $P[i_p\dd j_p]$ is a run in $P$ or $i_p=j_p=0$, and $T[i_t \dd j_t]$ is a run in $T$ or $i_t=j_t=0$. For $i_p,i_t > 0$, a block $D[i_p \dd j_p, i_t \dd j_t]$ is called \emph{homogeneous} if $P[i_p] = T[i_t]$. (For example, a block $D[3\dd 4][3\dd 6]$ in Fig.~\ref{fig:decreasing} is homogeneous.) A block such that all cells in it contain a value $q$, for some fixed integer $q$, is called a \emph{$q$-block}. (For example, a block $D[5 \dd 5][11\dd 14]$ in Fig.~\ref{fig:decreasing} is a $1$-block.) The \emph{border} of a block is the set of the cells contained in its top and bottom rows, as well as first and last columns. Consider a cell $(a,b)$ in $B$. We say that a block $B'$ is the \emph{top neighbour} of $B$ if it contains $(a-1,b)$, the \emph{left neighbour} if it contains $(a,b-1)$, and the \emph{diagonal neighbour} if it contains $(a-1,b-1)$. 

The following lemma is shown by induction in Appendix~\ref{sec:omitted}:

\begin{replemma}{nondecreasing}
\label{lm:non-decreasing}
Consider a block $B = D[i_p\dd j_p, i_t \dd j_t]$ and cell $(a,b)$ in it. If $i_p \leq a < j_p$, then $D[a,b] \le D[a+1,b]$ and if $i_t \leq b < j_t$, then $D[a,b] \le D[a,b+1]$.
\end{replemma}

By Lemma~\ref{lm:recursion}, inside a homogeneous block each value is equal to the minimum of its neighbours. Therefore, the values in a row or in a column cannot increase and we have the following corollary:
\begin{corollary}\label{cor:homogeneous}
Each homogeneous block is a $q$-block for some value $q$.
\end{corollary}


\section{Main result: $\Oh(kmn)$-time algorithm}
\label{sec:block}
In this section, we show Theorem~\ref{th:block} that for a pattern $P$ with $m$ runs and a text~$T$ with $n$ runs gives an $\Oh(kmn)$-time algorithm. We start with the following lemma which is a keystone to our result: 

\begin{lemma}\label{lm:border}
For a block $D[i_p \dd j_p, i_t \dd j_t]$ let $h=j_p-i_p$, $w =j_t-i_t$, and $d=d(P[i_p],T[i_t])$. We have for every $i_p < x \leq j_p$:
\begin{equation}\label{eq:border-right}
D[x,j_t]= \begin{cases}
D[i_p,j_t-(x-i_p)]+(x-i_p) \cdot d \text{ if } x-i_p \leq w; \\
D[x-w,i_t]+w \cdot d \text{ otherwise}.
\end{cases}
\end{equation}
For every $i_t < y \leq j_t$:
\begin{equation}\label{eq:border-bottom}
D[j_p,y]= \begin{cases}
D[j_p-(y-i_t),i_t]+(y-i_t) \cdot d \text{ if }  y-i_t \leq h; \\
D[i_p,y-h]+h \cdot d \text{ otherwise}.
\end{cases}
\end{equation}
\end{lemma}
\begin{proof}
For a homogeneous block, we have $d=0$, and by Corollary~\ref{cor:homogeneous} all the values in such a block are equal, hence the claim of the lemma is trivially true. 

Assume now $d > 0$. Consider $x$, $i_p < x \leq j_p$, and let us show Eq.~\ref{eq:border-right}, Eq.~\ref{eq:border-bottom} can be shown analogously. Let $\pi$ be a warping path realizing $D[x,j_t]$.
Let $(a,b)$ be the first node of $\pi$ belonging to the block.
We have $a \in [i_p,j_p]$ and $b \in [i_t,j_t]$ and either $a=i_p$ or $b=i_t$.
The number of edges of $\pi$ in the block from $(a,b)$ to $(x,j_t)$ must be minimal, else there would be a shorter path, thus it is equal to $\max\{x-a,j_t-b\}$ and $D[x,j_t]=D[a,b]+\max\{x-a,j_t-b\}\cdot d$.

\def\height{4}
\def\width{3}
\def\cell{0.5}

\newcommand{\Block}{
\node (bl) at (0,0) {};
\node (br) at ($(bl)+(\width,0)$) {};
\node (tl) at ($(bl)+(0,\height)$) {};
\node[below,xshift=-5pt] () at ($(tl)$) {$i_p$};
\node (tr) at ($(bl)+(\width,\height)$) {};
\node[xshift=-15pt] () at ($(tr)+(-\width,-\width)$)  {$i_p+w$};
\draw ($(bl)$) rectangle ($(tr)$);
\draw[dashed] ($(tr)+(0,-\width)$) -- ($(tr)+(-\width,-\width)$) {};
}

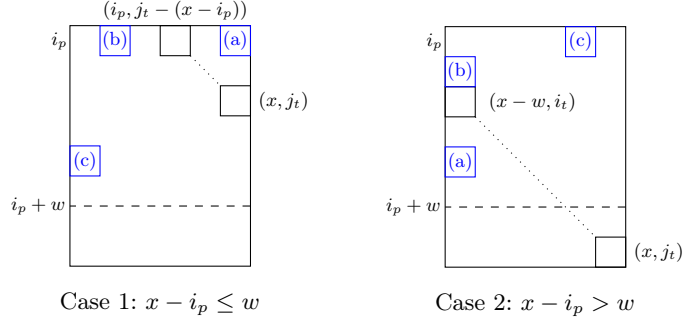
\begin{figure}
\centering

\begin{subfigure}{0.4\textwidth}
\centering
\begin{tikzpicture}[scale=0.8, every node/.style={scale=0.8}]
\Block
\node (r1) at ($(tr)+(-1.5,0)$) {};
\node (r2) at ($(tr)+(-\cell,-1.5+\cell)$) {};
\draw ($(r1)$) rectangle ($(r1)+(\cell,-\cell)$)node [midway,yshift=0.5cm] {$(i_p,j_t-(x-i_p))$};
\draw ($(r2)$) rectangle ($(r2)+(\cell,-\cell)$)node [midway,xshift=0.8cm] {$(x,j_t)$};
\draw[dotted] ($(r1)+(\cell,-\cell)$) -- ($(r2)$);

\node (ra) at ($(tr)+(-\cell,0)$) {};
\draw[blue] ($(ra)$) rectangle ($(ra)+(\cell,-\cell)$)node [midway] {(a)};
\node (rb) at ($(tl)+(0.5,0)$) {};
\draw[blue] ($(rb)$) rectangle ($(rb)+(\cell,-\cell)$)node [midway] {(b)};
\node (rc) at ($(tl)+(0,-2)$) {};
\draw[blue] ($(rc)$) rectangle ($(rc)+(\cell,-\cell)$)node [midway] {(c)};
\end{tikzpicture}
\caption*{Case 1:  $x-i_p \leq w$}
\label{fig:case1}
\end{subfigure}
\begin{subfigure}{0.4\textwidth}
\centering
\begin{tikzpicture}[scale=0.8, every node/.style={scale=0.8},]
\Block
\node (r1) at ($(bl)+(0,3)$) {};
\node (r2) at ($(r1)+(+\width-\cell,-\width+\cell)$) {};
\draw ($(r1)$) rectangle ($(r1)+(\cell,-\cell)$)node [midway,xshift=1.2cm] {$(x-w,i_t)$};
\draw ($(r2)$) rectangle ($(r2)+(\cell,-\cell)$)node [midway,xshift=0.8cm] {$(x,j_t)$};
\draw[dotted] ($(r1)+(\cell,-\cell)$) -- ($(r2)$);
\path ($(tl)+(0.5,0.55)$) rectangle ($(tl)+(\cell,-\cell)$);

\node (ra) at ($(bl)+(0,2)$) {};
\draw[blue] ($(ra)$) rectangle ($(ra)+(\cell,-\cell)$)node [midway] {(a)};
\node (rb) at ($(tl)+(0,-0.5)$) {};
\draw[blue] ($(rb)$) rectangle ($(rb)+(\cell,-\cell)$)node [midway] {(b)};
\node (rc) at ($(tl)+(2,0)$) {};
\draw[blue] ($(rc)$) rectangle ($(rc)+(\cell,-\cell)$)node [midway] {(c)};
\end{tikzpicture}
\caption*{Case 2: $x-i_p > w$}
\label{fig:case2}
\end{subfigure}

\caption{Cases of Lemma~\ref{lm:border}. Possible locations of the cell $(a,b)$ are shown in blue.}\label{fig:border}
\end{figure}

\noindent
\underline{Case 1: $x-i_p \leq w$.} Consider a cell $(i_p,j_t - (x-i_p))$. There is a path from $(i_p,j_t - (x-i_p))$ to $(x,j_t)$ that takes $x-i_p$ diagonal steps inside the block, and therefore $D[x,j_t] \leq D[i_p,j_t - (x-i_p)]+(x-i_p)\cdot d$. We now show that $D[x,j_t] \geq D[i_p,j_t - (x-i_p)]+(x-i_p)\cdot d$, which implies the claim of the lemma.
\begin{enumerate}[label=(\alph*)]
\item If \underline{$a=i_p$ and $b \geq j_t - (x-i_p)$}, then $\max\{x-i_p,j_t-b\}=x-i_p$. We have $D[x,j_t] =D[i_p,b]+(x-i_p)\cdot d  \geq D[i_p,j_t - (x-i_p)]+(x-i_p)\cdot d \text{ (Lemma~\ref{lm:non-decreasing})}$.

\item If \underline{$a=i_p$ and $b < j_t - (x-i_p)$}, then $\max\{x-i_p,j_t-b\}=j_t-b$. As there is a path from $(a,b) = (i_p,b)$ to $(i_p,j_t - (x-i_p))$ of length $(j_t - (x-i_p)-b)$, we have $D[i_p,j_t - (x-i_p)] \le D[i_p,b] + (j_t - (x-i_p)-b) \cdot d$. Consequently,
\begin{align*}
D[x,j_t]&=D[i_p,b]+(j_t-b)\cdot d \\
& \geq D[i_p,j_t - (x-i_p)] - (j_t - (x-i_p)-b) \cdot d +(j_t-b) \cdot d \text{ (Lem.~\ref{lm:recursion})}\\
& = D[i_p,j_t - (x-i_p)]+(x-i_p)\cdot d
\end{align*}
\item If \underline{$b=i_t$}, then $i_p \leq a$ and $\max\{x-a, j_t-b\} \leq \max\{x-i_p,w\} = w$. As there is a path from $(i_p,i_t)$ to $(i_p,j_t - (x-i_p))$ of length $(j_t - (x-i_p)-i_t)$, we have $D[i_p,j_t - (x-i_p)] \le D[i_p,i_t] + (j_t - (x-i_p)-i_t) \cdot d$. Therefore, 
\begin{align*}
D[x,j_t]&=D[a,i_t]+w\cdot d \geq D[i_p,i_t] +w\cdot d \text{ (Lemma~\ref{lm:non-decreasing})} \\
& \geq D[i_p,j_t - (x-i_p)] - (j_t - (x-i_p)-i_t) \cdot d +w\cdot d\\
& = D[i_p,j_t - (x-i_p)]+(x-i_p)\cdot d
\end{align*}
\end{enumerate}

\noindent
\underline{Case 2: $x-i_p > w$.} Consider a cell $(x-w,i_t)$. There is a path from $(x-w,i_t)$ to $(x,j_t)$ that takes $w$ diagonal steps inside the block, and therefore $D[x,j_t] \leq D[x-w,i_t]+w \cdot d$. We now show that $D[x,j_t] \geq D[x-w,i_t]+w\cdot d$, which implies the claim of the lemma.
\begin{enumerate}[label=(\alph*)]

\item If \underline{$b=i_t$ and $a \ge x-w$},
then $\max\{x-a,j_t-b\} = \max\{x-a,w\}=w$ and we have $D[x,j_t] =D[a,i_t]+w\cdot d \geq D[x-w,i_t] +w\cdot d \text{ (Lemma~\ref{lm:non-decreasing})}$. 

\item  If \underline{$b=i_t$ and $a < x-w$}, then $\max\{x-a,j_t-b\}= \max\{x-a,w\}=x-a$. As there is a path from $(a,i_t)$ to $(x-w,i_t)$ of length $(x-w-a)$, we have $D[x-w,i_t] \leq D[a,i_t] + (x-w-a) \cdot d$ by definition. Therefore, 
\begin{align*}
D[x,j_t]&=D[a,i_t]+(x-a)\cdot d\\
& \geq D[x-w,i_t] - (x-w-a) \cdot d+(x-a) \cdot d\\
& = D[x-w,i_t] + w\cdot d
\end{align*}
\item If \underline{$a=i_p$}, 
$b \geq i_t$ and thus $\max\{x-a,j_t-b\} \leq \max\{x-i_p,w\} = x-i_p$. Additionally, as there is a path from $(i_p,i_t)$ to $(x-w,i_t)$ of length $(x-w - i_p)$ we have $D[x-w,i_t] \le D[i_p,i_t] + (x-w-i_p) \cdot d$. Consequently,
\begin{align*}
D[x,j_t]&=D[i_p,b]+(x-i_p)\cdot d \geq D[i_p,i_t] +(x-i_p)\cdot d \text{ (Lemma~\ref{lm:non-decreasing})} \\
& \geq D[x-w,i_t] - (x-w-i_p) \cdot d +(x-i_p)\cdot d\\
& = D[x-w,i_t] + w\cdot d
\end{align*}
\end{enumerate}
\qed
\end{proof}

We say that a cell in a border of a block is \emph{interesting} if its value is at most $k$. To solve the $k$-$\dtw$ problem it suffices to compute the values of all interesting cells in the last row of $D$. Consider a block $B = D[i_p\dd j_p, i_t \dd j_t]$ and recall that the values in it are non-decreasing top to down and left to right (Lemma~\ref{lm:non-decreasing}). We can consider the following compact representation of its interesting cells. For an integer $\ell$, define $q_{\text{top}}^\ell \in [i_t,j_t]$ to be the last position such that $D[i_p,q_{\text{top}}^\ell] \le \ell$, and $q_{\text{bot}}^\ell \in [i_t,j_t]$ the last position such that $D[j_p,q_{\text{bot}}^\ell] \le \ell$. If a value is not defined, we set it equal to $i_t-1$. Analogously, define $q_{\text{left}}^\ell \in [i_p,j_p]$ to be the last position such that $D[q_{\text{left}}^\ell,i_t] \le \ell$, and $q_{\text{right}}^\ell \in [i_p,j_p]$ the last position such that $D[q_{\text{right}}^\ell,j_t] \le \ell$. If a value is not defined, we set it equal to $i_p-1$.  Positions $q_{\text{top}}^0, \ldots, q_{\text{top}}^k$ uniquely describe the interesting border cells in the top row of $B$, $q_{\text{bot}}^0, \ldots, q_{\text{bot}}^k$ in the bottom row, $q_{\text{left}}^0, \ldots, q_{\text{left}}^k$ in the leftmost column,  $q_{\text{right}}^0, \ldots, q_{\text{right}}^k$ in the rightmost column. 

\begin{lemma}\label{lm:top-left}
The compact representations of the interesting border cells in the top row and the leftmost column of a block $B$ can be computed in $\Oh(k)$ time given the compact representation of the interesting border cells in its neighbours.
\end{lemma}
\begin{proof}
We explain how to compute the representation for the leftmost column of $B$, the representation for the top row is computed analogously.  Let $d = d(P[i_p],T[i_t])$. If $d=0$ (the block is homogeneous), by Corollary~\ref{cor:homogeneous} the block is a $q$-block for some value $q$ which can be computed in $\Oh(1)$ time by Lemma~\ref{lm:recursion} if it is interesting (and otherwise we have a certificate that the value is not interesting). We can then derive the values $q_{\text{left}}^\ell$, $\ell = 0, 1, \ldots, k$ in $\Oh(k)$ time.

Assume now $d > 0$. We start by computing $D[i_p,i_t]$ using Lemma~\ref{lm:recursion}. We note that if $D[i_p,i_t] \le k$, then we know the values of its neighbours realising it and therefore can compute it, otherwise we can certify that $D[i_p,i_t] > k$. Assume $D[i_p,i_t] = v$, which implies that $q_{\text{left}}^0, \ldots, q_{\text{left}}^{\min\{k,v\}-1}$ equal $i_p-1$. We must now compute $q_{\text{left}}^{\min\{k,v\}}, \ldots, q_{\text{left}}^{k}$. Consider a cell $(q,i_t)$ of the block with $q > i_p$.  The second to the last cell in the warping path that realizes $D[q,i_t] = \ell$ is one of the cells $(q-1,i_t)$, $(q-1,i_t-1)$ or $(q,i_t-1)$, and the value of the path up to there must be $\ell-d$. Note that all the three cells belong either to the leftmost column of $B$, or the rightmost column of its left neighbour. Consequently, for all $\min\{k,v\} < \ell \le k$, we have $q_{\text{left}}^\ell = \min\{\max\{q_{\text{left}}^{\ell-d}, r_{\text{right}}^{\ell-d}\} + 1\}, j_t\}$, and the positions $q_{\text{left}}^0, \ldots, q_{\text{left}}^{k}$ can be computed in $\Oh(k)$ time.
\qed
\end{proof}

\begin{lemma}\label{lm:bottom-right}
The compact representations of the interesting border cells in the bottom row and the rightmost column of a block $B$ can be computed in $\Oh(k)$ time given the compact representation of the interesting border cells in its leftmost column and the top row.
\end{lemma}
\begin{proof}
We explain how to compute the representation for the bottom row, the representation for the rightmost column is computed analogously. 

Eq.~\ref{eq:border-bottom} and the compact representations of the leftmost column and the top row of $B$ partition the bottom row of $B$ into $\Oh(k)$ intervals (some intervals can be empty), and in each interval the values are described either as a constant or as a linear function. (See Fig.~\ref{fig:borders_represent}.) Formally, let $h=j_p-i_p$. By Eq.~\ref{eq:border-bottom}, for $y \in [i_t,  j_p+i_t-q^{k}_{\text{left}}-1] \cap [i_t,j_t]$ we have $D[j_p][y] > k$. For $y \in [j_p+i_t-q^\ell_\text{left}, j_p+i_t-q^{\ell-1}_\text{left}-1]  \cap [i_t,j_t]$, $\ell = k, k-1, \ldots, 1$, we have 
$D[j_p][y] = \ell + (y-i_t) \cdot d$. For $y \in [j_p+i_t-q^0_\text{left}, j_p+i_t-i_p] \cap [i_t,j_t]$ we have $D[j_p][y] = (y-i_t) \cdot d$. For $y \in [i_t+h, q_{\text{top}}^{0}+h-1] \cap [i_t,j_t]$ we have $D[j_p][y] = h \cdot d$. For $y \in [q_{\text{top}}^{\ell}+h,q_{\text{top}}^{\ell+1}+h-1] \cap [i_t,j_t], \ell = 0, 1, \ldots, k-1$, we have $D[j_p][y] = \ell + h \cdot d$. Finally, for $y \in [q_{\text{top}}^{k}+h, j_t]$, there is $D[j_p][y] > k$ again.

\def\height{4}
\def\h{4.25}
\def\width{10}
\def\cell{0.5}

\newcommand{\Block}{
\node (bl) at (0,0) {};
\node (br) at ($(bl)+(\width,0)$) {};
\node (tl) at ($(bl)+(0,\height)$) {};
\node (tr) at ($(bl)+(\width,\height)$) {};
\draw ($(bl)$) rectangle ($(tr)$);
\node[yshift=-5pt,xshift=-5pt] () at ($(tl)$) {\tiny{$i_p$}};
\node[yshift=5pt,xshift=-5pt] () at ($(bl)$) {\tiny{$j_p$}};
\node[yshift=7.5pt,xshift=5pt] () at ($(tl)$) {\tiny{$i_t$}};
\node[yshift=8pt,xshift=-5pt] () at ($(tr)$) {\tiny{$j_t$}};
}

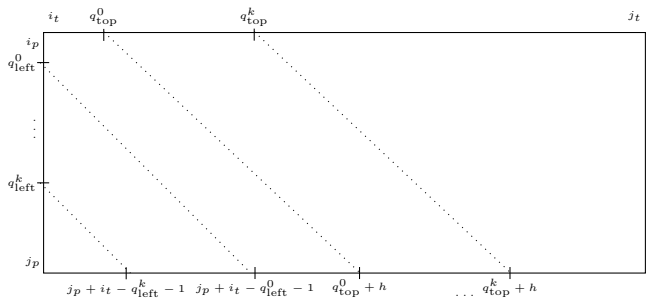
\begin{figure}
\centering

\begin{tikzpicture}[scale=0.8, every node/.style={scale=0.8}]
\Block
\node[left] (r1) at ($(tl)+(0,-.5)$) {\tiny{$q_{\text{left}}^0$}};
\draw ($(r1)+(0.25,0)$)--($(r1)+(0.45,0)$);
\node[left] (r3) at ($(tl)+(0,-1.5)$) {\tiny{$\vdots$}};
\node[left] (r4) at ($(tl)+(0,-2.5)$) {\tiny{$q_{\text{left}}^k$}};
\draw ($(r4)+(0.25,0)$)--($(r4)+(0.45,0)$);

\draw[dotted] ($(r1)+(0.3,0)$)--($(r1)+(\height-0.2,-\height+0.5)$);
\draw[dotted] ($(r4)+(0.3,0)$)--($(r4)+(\height-2.2,-\height+2.5)$);

\node[below,xshift=-2pt] (r5) at ($(r4)+(\height-2.2,-\height+2.5)$) {\tiny{$j_p+i_t-q^k_{\text{left}}-1$}};
\draw ($(r5)+(0,0.15)$)--($(r5)+(0,0.35)$);
\node[below,xshift=2pt] (r8) at ($(r1)+(\height-0.2,-\height+0.5)$) {\tiny{$j_p+i_t-q^0_{\text{left}}-1$}};
\draw ($(r8)+(0,0.15)$)--($(r8)+(0,0.35)$);

\node[above] (h1) at ($(tl)+(1,0)$) {\tiny{$q_{\text{top}}^0$}};
\draw ($(h1)+(0,-0.2)$)--($(h1)+(0,-0.4)$);
\node[above] (h4) at ($(tl)+(3.5,0)$) {\tiny{$q_{\text{top}}^k$}};
\draw ($(h4)+(0,-0.2)$)--($(h4)+(0,-0.4)$);

\draw[dotted] ($(h1)+(0,-0.25)$)--($(h1)+(\h,-\h)$);
\draw[dotted] ($(h4)+(0,-0.25)$)--($(h4)+(\h,-\h)$);
                                                                                                                                                                                                                                                                                                                                                                                                                                                                                                                                                                                                                                                                                                                                                                                                                                                        
\node[below] (h5) at ($(h1)+(\h,-\h)$) {\tiny{$q_{\text{top}}^0+h$}};
\draw ($(h5)+(0,0.15)$)--($(h5)+(0,0.35)$);
\node[below] (h7) at ($(bl)+(7,-0.25)$) {\tiny{$\ldots$}};
\node[below] (h8) at ($(h4)+(\h,-\h)$) {\tiny{$q_{\text{top}}^k+h$}};
\draw ($(h8)+(0,0.15)$)--($(h8)+(0,0.35)$);
\end{tikzpicture}

\caption{Compressed representation of interesting border cells.}\label{fig:borders_represent}
\end{figure}

By Lemma~\ref{lm:non-decreasing}, the values in the bottom row are non-decreasing. We scan the intervals from left to right to compute the values $q_{\text{bot}}^0, \ldots, q_{\text{bot}}^k$ in $\Oh(k)$ time. In more detail, let $q_{\text{bot}}^\ell$ be the last computed value, and $[i,j]$ be the next interval. We set $q_{\text{bot}}^{\ell+1} = q_{\text{bot}}^{\ell}$. If the values in the interval are constant and larger than $\ell+1$, we continue to computing $q_{\text{bot}}^{\ell+2}$. If the values are increasing linearly, we find the position of the last value smaller or equal to $\ell+1$, set $q_{\text{bot}}^{\ell+1}$ equal to this position, and continue to computing $q_{\text{bot}}^{\ell+2}$. Finally, if the values in the interval are constant and equal to $\ell+1$, we update $q_{\text{bot}}^{\ell+1} = j$ and continue to the next interval. As soon as $q_{\text{bot}}^k$ is computed, we stop the computation. 
\qed
\end{proof}

Since there are $\Oh(mn)$ blocks in total, Lemmas~\ref{lm:top-left} and  \ref{lm:bottom-right} immediately imply Theorem~\ref{th:block}.


\section{Approximation algorithm}
\label{sec:approx}
In this section, we show an approximation algorithm for computing the smallest $\dtw$ distance between a pattern $P$ and a substring of a text $T$. We assume that the $\dtw$ distance is defined over a metric on the alphabet $\Sigma$. Kuszmaul~\cite{DBLP:conf/icalp/Kuszmaul19} showed that the problem of computing the smallest $\dtw$ distance over an arbitrary metric can be reduced to the problem of computing the smallest distance over a so-called well-separated tree metric: 

\begin{definition}[Well-separated tree metric]
 Consider a rooted tree $\tau$ with positive weights on the edges whose leaves form an alphabet $\Sigma$. The tree $\tau$ specifies a metric $\mu_\tau$ on $\Sigma$: The  \emph{distance} between two leaves $a,b \in \Sigma$ is defined as the maximum weight of an edge in the shortest path from $a$ to $b$. The metric $\mu_\tau$ is a \emph{well-separated tree metric} if the weights of the edges are not increasing in every root-to-leaf path.  The \emph{depth} of $\mu_\tau$ is defined to be the depth of $\tau$.
\end{definition}

Below we show that Theorem~\ref{th:block} implies the following result for well-separated tree metrics:

\begin{lemma}\label{lm:approx}
Given run-length encodings of a pattern $P$ with $m$ runs and a text~$T$ with $n$ runs over an alphabet $\Sigma$. Assume that the $\dtw$ distance is specified by a well-separated tree metric $\mu_\tau$ on $\Sigma$ with depth $h$, and suppose that the ratio between the largest and the smallest non-zero distances between the letters of $\Sigma$ is at most exponential in $L = \max\{|P|,|T|\}$. For any $0 < \epsilon < 1$, there is an $\Oh (L^{1-\eps} \cdot hmn \log L)$-time algorithm that computes $\Oh(L^{\eps})$-approximation of the smallest $\dtw$ distance between $P$ and a substring of $T$.
\end{lemma}

By plugging the lemma into the framework of~\cite{DBLP:conf/icalp/Kuszmaul19}, we obtain:

\begin{reptheorem}{approx}
\label{th:approx}
Given run-length encodings of a pattern~$P$ with $m$ runs and of a text $T$ with $n$ runs over an alphabet $\Sigma$. Assume that the $\dtw$ distance is specified by a metric $\mu$ on $\Sigma$, and suppose that the ratio between the largest and the smallest non-zero distances between the letters of $\Sigma$ is at most exponential in $L = \max\{|P|,|T|\}$.
For any $0 < \epsilon < 1$, there is a $\Oh(L^{1-\eps} \cdot mn \log^3 L)$-time algorithm that computes $\Oh(L^{\eps})$-approximation of the smallest $\dtw$ distance between $P$ and a substring of $T$ correctly with high probability\footnote{The preprocessing time $\Oh(|\Sigma|^2 \log L)$ that is required to embed $\mu$ into a well-separated metric is not accounted for in the runtime of the algorithm.}.
\end{reptheorem}

The proof follows the lines of the full version~\cite{https://doi.org/10.48550/arxiv.1904.09690} of~\cite{DBLP:conf/icalp/Kuszmaul19}, we provide it in Appendix~\ref{app:approximate} for completeness. We now show Lemma~\ref{lm:approx}. Compared to~\cite{DBLP:conf/icalp/Kuszmaul19}, the main technical challenge is that our $k$-$\dtw$ algorithm (Theorem~\ref{th:block}) assumes an integer-valued distance function on the alphabet. We overcome this by developing an intermediary $2$-approximation algorithm for real-valued distances (see the two claims below). 

\paragraph{Proof of Lemma~\ref{lm:approx}.} For brevity, let $\delta$ be the smallest $\dtw_{\mu_\tau}$ distance between $P$ and a substring of $T$.

\begin{claim}\label{claim:dtw_large}
Let $0 < \eps < 1$. Assume that for all $a,b \in \Sigma$, $a \neq b$, there is $\mu_\tau(a,b) \ge \gamma$ and that the value of $\mu_\tau(a,b)$ can be evaluated in $\Oh(t)$ time. There is an $\Oh(L^{1-\eps} tmn)$-time algorithm which either computes a $2$-approximation of $\delta$ or concludes that it is larger than $\gamma \cdot L^{1-\eps}$. 
\end{claim}
\begin{proof}
Define a new distance function $\mu_\tau'(a,b) = \lceil \mu_\tau(a,b)/\gamma \rceil$. For all $a,b \in \Sigma$, $a \neq b$, we have $\mu_\tau(a,b) \leq \gamma \cdot \mu_\tau'(a,b) \le \mu_\tau(a,b) + \gamma \le 2 \mu_\tau(a,b)$. 
Consequently, for all strings $X,Y$ we have $\dtw_{\mu_\tau}(X,Y) \le \gamma \cdot \dtw_{\mu_\tau'}(X,Y) \le 2 \dtw_{\mu_\tau}(X,Y)$. 
Let $\delta' = \min_{S-\text{ substring of } T} \min\{2k+1, \dtw_{\mu_\tau'}(P,S)\}$ for $k = L^{1 - \eps}$. By Theorem~\ref{th:block}, it can be computed in $\Oh(L^{1-\eps} tmn)$ time. If $\delta' = 2L^{1 - \eps}+1$, we conclude that $\delta \geq \gamma \cdot L^{1 - \eps}$, and otherwise, output $\gamma \delta'$.
\qed\end{proof}

W.l.o.g., the minimum non-zero distance between two distinct letters of $\Sigma$ is~$1$ and the largest distance is some value $M$, which is at most exponential
in $L$. We run the algorithm above for $\gamma = 1$, which either computes a $2$-approximation of $\delta$ which we can output immediately, or concludes that $\delta \ge L^{1 -\eps}$. Below we assume that $\delta \ge L^{1 -\eps}$. 

\begin{definition}[$r$-simplification]
For a string $X \in \Sigma^\ast$ and $r \ge 1$, the
\emph{$r$-simplification} $s_r(X)$ is constructed by replacing
each letter $a$ of $X$ with its highest ancestor $a'$ in $\tau$ that can
be reached from $a$ using only edges of weight $\le r / 4$.
\end{definition}

\begin{fact}[{Corollary of~\cite[Lemma 4.6]{DBLP:conf/icalp/Kuszmaul19}, see also~\cite{DBLP:conf/compgeom/BravermanCKWY19}}]\label{fact:simplified}
For all $X,Y \in \Sigma^{\le L}$, the following properties hold:
\begin{enumerate}
  \item $\dtw_{\mu_\tau}(s_r(X), s_r(Y)) \le \dtw_{\mu_\tau}(X,Y)$.
  \item If $\dtw_{\mu_\tau}(X,Y) > L r$, then $\dtw_{\mu_\tau}(s_r(X), s_r(Y)) > L r/2$. 
\end{enumerate}
\end{fact}

Fix $r \ge 1$ and $0 < \eps < 1$. In the \emph{$(L^\eps, r)$-$\dtw$ gap pattern matching problem}, we must output $0$ if the smallest $\dtw$ distance between $P$ and a substring of~$T$ is at most $L^{1 - \eps} r/4$ and $1$ if it is at least $L r$, otherwise we can output either $0$ or $1$. 

\begin{claim}\label{lm:tree_metric}
The $(L^{\eps}, r)$-$\dtw$ gap pattern matching problem can be solved in $\Oh(L^{1-\eps} \cdot hmn)$ time. 
\end{claim}
\begin{proof}
Let $\delta_r$ be the smallest $\dtw_{\mu_\tau}$ distance between $s_r(P)$ and a substring of $s_r(T)$. 
If $L^{1-\eps} > L/2$, then $L = \Oh(1)$ and we can compute $\delta$ exactly in $\Oh(1)$ time by Lemma~\ref{lm:recursion}. Otherwise, we run the $2$-approximation algorithm for $\gamma = r/4$, which takes $O(L^{1 - \eps} \cdot hmn)$ time (we can evaluate the distance between two letters in $\Oh(h)$ time). If the algorithm concludes that $\delta_r > L^{1 - \eps} r /4$, then $\delta >  L^{1 - \eps} r /4$ by Fact~\ref{fact:simplified}, and we can output $1$. Otherwise, the algorithm outputs a $2$-approximation $\delta_r'$ of $\delta_r$, i.e. $\delta_r \le \delta_r' \le 2\delta_r$. If $\delta_r' \le L^{1 - \eps} r \le Lr / 2$, then we have $\delta_r \le Lr / 2$. Therefore, $\delta \le Lr$ by Fact~\ref{fact:simplified} and we can output~$0$. Otherwise, $\delta \ge \delta_r \ge \delta_r'/2 > L^{1 - \eps} r/2 > L^{1 - \eps} r/4$, and we can output~$1$.  
\qed\end{proof}

Consider the $(L^{\eps} / 2, 2^i)$-$\dtw$ gap pattern matching problem for $0 \le i \le \lceil \log ML \rceil$. If the $(L^{\eps} / 2, 2^0)$-$\dtw$
gap pattern matching problem returns $0$, then we know that $\delta \le L$,
and can return $L^{1 - \eps}$ as a $L^{\eps}$-approximation for~$\delta$. Therefore, it suffices to consider the case where the $(L^{\eps} / 2, 2^0)$-$\dtw$ gap pattern matching problem returns $1$. We can assume, without computing it, that the $(L^{\eps}/2, 2^{ \lceil
  \log ML \rceil})$-$\dtw$ gap pattern matching returns $0$ as $\delta \le M L$. Consequently, there must exist $i^\ast$ such that $(L^{\eps} /
2, 2^{i^\ast - 1})$-$\dtw$ gap pattern matching returns $1$ and $(L^{\eps} / 2, 2^{i^\ast - 1})$-$\dtw$ returns $0$. We can find $i^\ast$ by a binary
search which takes $\Oh(L^{1 - \eps} hmn \log \log M L) = \Oh(L^{1 - \eps} hmn \log L)$ time. We have $\delta \ge  2^{i^\ast-1} L^{1 - \eps} / 4$ and $\delta \le 2^{i^\ast} L$, and therefore can return $2^{i^\ast-1} L^{1 - \eps} / 4$ as a
$\Oh(L^{\eps})$-approximation of $\delta$.
\qed


\section{Experiments}\label{sec:experiments}
This section provides evidence of the advantage of the $\dtw$ distance over the edit distance when processing the third generation sequencing (TGS) data. Our experiment compares how the two distances are affected by biological mutation as opposed to sequencing errors, including homopolymer length errors. 

We first simulate two genomes, $G$ and $G'$, which can be considered as strings on the alphabet $\Sigma = \{A,C,G,T\}$. The genome $G$ is a substring of the E.coli genome (strain SQ110, NCBI Reference Sequence: NZ\_CP011322.1) of length $10000$ (positions $100000$ to $110000$, excluded). The genome $G'$ is obtained from $G$ by simulating biological mutations, where the probabilities are chosen according to~\cite{10.1093/molbev/msp063}. The algorithm initializes $G'$ as the empty string, and $\texttt{pos} = 1$. While $\texttt{pos} \le |G|$ it executes the following:

\begin{enumerate}
	\item With probability $0.01$, simulate a substitution: chose  uniformly at random  $a \in \Sigma$, $a \neq G[pos]$. Set $G' = G'a$ and $\texttt{pos} = \texttt{pos} + 1$. 
	\item Else, with probability $0.0005$ simulate an insertion or a deletion of a substring of length $x$, where $x$ is chosen uniformly at random from an interval $[1, \texttt{max\_len\_ID}]$, where $\texttt{max\_len\_ID}$ is fixed to $10$ in the experiments:
		\begin{enumerate}
			\item With probability $0.5$, set $\texttt{pos} = \texttt{pos} + x + 1$ (deletion);
			\item With probability $0.5$, choose a string $X \in \Sigma^x$ uniformly at random, set $G' = G'X$ and $\texttt{pos}= \texttt{pos}+1$ (insertion).  
		\end{enumerate}
	\item Else, set $G' = G'G[\texttt{pos}]$ and $\texttt{pos} = \texttt{pos} + 1$. 
\end{enumerate}

\noindent To simulate reads, we extract substrings of $G'$ and add sequencing errors: 

\begin{enumerate} 
	\item For each read, extract a substring $R$ of length $500$ at a random position of~$G'$. As $G'$ originates from $G$, we know the theoretical distance from $R$ to $G$, which we call the ``\emph{biological diversity}''. The biological diversity is computed as the sum of the number of letter substitutions, letter insertions, and letter deletions that were applied to the original substring from $G$ to obtain $R$. 
	\item Add sequencing errors by executing the following for each position $i$ of~$R$:
	\begin{enumerate}
		\item With probability $0.001$, substitute $R[i]$ with a letter $a \in \Sigma$, $a \neq R[i]$. The letter $a$ is chosen uniformly at random.  
		\item If $R[i] = R[i-1]$, insert with a  probability $p_{hom}$ a third occurrence of the same letter to simulate a homopolymer error.
	\end{enumerate}
\end{enumerate}

Fig.~\ref{fig:experiments_N_600_max_indel_length_10} shows the difference between the biological diversity and the smallest edit and $\dtw$ distances between a generated read and a substring of~$G$ depending on $p_{hom}$. It can be seen that the $\dtw$ distance gives a good estimation of the biological diversity, whereas, as expected, the edit distance is heavily affected by homopolymer errors.  To ensure reproducibility of our results, our complete experimental setup is available at \url{https://github.com/fnareoh/DTW}.


\begin{figure}
\centering
\includegraphics[scale=0.5]{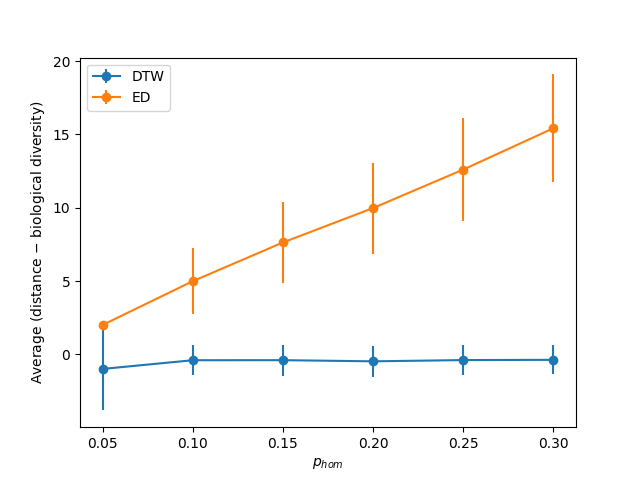}
\caption{Edit and $\dtw$ distances offset by the biological diversity as a function of $p_{hom}$. Each point is averaged over 600 reads ($\times 30$ coverage).}
\label{fig:experiments_N_600_max_indel_length_10}
\end{figure}


\bibliographystyle{splncs04}
\bibliography{main}

\begin{subappendices}
\renewcommand{\thesection}{\Alph{section}}%
\section{}
\label{sec:k=1}
In this section, we show Lemma~\ref{lm:1-DTW} that for a pattern $P$ with $m$ runs and and text $T$ with $n$ runs gives an $\Oh(m+n)$-time algorithm.  

\begin{definition}[RLE-diagonals]
We say that a sequence of blocks forms an \emph{RLE-diagonal} if the blocks are formed by runs $i, i+1, \ldots, j$ of $P$ and $i+\delta, i+1+\delta, \ldots, j+\delta$ of $T$, for some integers $i,j,\delta$. 
\end{definition}

\begin{definition}[Streak]
A $q$-streak is a maximal subsequence of an RLE-diagonal containing sequential homogeneous $q$-blocks. 
\end{definition}

\begin{observation}\label{obs:zero_streak}
If $D[i,j] = 0$, then it belongs to a $0$-streak. Furthermore, each $0$-streak necessarily starts in the first row of $D$.
\end{observation}
\begin{proof}
By definition, there must be a path from the first row of $D$ to $D[i,j]$ containing $0$-values only. For every $0$-value $D[i',j']$ we must have $P[i'] = T[j']$, and therefore every such value must belong to a homogeneous $0$-block. Furthermore, two homogeneous blocks can only be neighbours diagonally, else it would contradict the maximality of the runs. The claim follows. 
\qed
\end{proof}

\begin{observation}\label{obs:one_value}
If $D[i,j] = 1$, then $D[i,j]$ belongs to a $1$-streak or neighbours a block in a $0$-streak.
\end{observation}
\begin{proof}
If $P[i]=T[j]$, we are in a homogeneous block and $D[i,j]$ belongs to a $1$-streak, and we are done. Otherwise, we have $P[i] \neq T[j]$ and there is a path $(i_1,j_1), (i_2,j_2), \ldots, (i_q,j_q)$ such that $i_1 = 1$, $(i_q,j_q) = (i,j)$, and $D[i_{q},j_{q}] = \sum_{q'=1}^{q} d(P[i_{q'}], T[j_{q'}])$. As $d(P[i_q],T[i_q]) \ge 1$, it follows that $d(P[i_q],T[i_q]) = 1$ for all $1 \le q' \le q-1$, $d(P[i_{q'}], T[j_{q'}]) = 0$, and therefore $D[i_{q'}, j_{q'}]$ must belong to a $0$-streak by Observation~\ref{obs:zero_streak}. 
\qed
\end{proof}

\repeatlemma{1DTW}
\begin{proof}
For a string $S$, define $\overline{RLE}(S)$ to be a string such that $\overline{RLE}(S)[i]$ contains the letter forming the $i$-th run of $S$. For example, $\overline{RLE}(aabbbc) = abc$. We preprocess $P' = \overline{\RLE}(P)$ and $T' = \overline{\RLE}(T)$ in $\Oh(m+n)$ time and space to maintain longest common suffix queries in constant time~\cite{RMQ}. The input of a longest common suffix query are two positions $i,j$ of $P'$ and $T'$ respectively, and the output is the largest $\ell$ such that $P'[i-\ell \dd i] = T'[j-\ell \dd j]$.

Let $B_i$, $1 \le i \le n$, be the block of $D$ formed by the $m$-th run in $P$ and the $i$-th run in $T$. Using one longest common suffix query for each block $B_i$, we find the maximal streak containing it. If this streak reaches the first row of $D$, it is a $0$-streak and all the values in the bottom row of $B_i$ are zeros. 

We must now decide which entries in the $M$-th row of $D$ must be filled with one. Consider an entry $D[M,\ell] \neq 0$ that belongs to a block $B_i$. 

If $B_i$ is contained in a streak of length at least one, then for $D[M,\ell]$ to be equal to one, it must be a $1$-streak. Consider the first block in the maximal  streak containing $B_i$, and let $c$ be the cell in its top left corner. Because $c$ can not be equal to zero, it suffices to check whether the value in $c$ equals one. Consider a path realizing the value of $c$. It goes either through the left neighbour $\ell$ of~$c$, the top neighbour $t$ of $c$, or the diagonal neighbour $d$ of $c$. Furthermore, the value in $c$ equals the minimum of the values in $\ell, d, t$. Therefore, the value in $c$ equals one iff one of the values in $\ell, d, t$ equals one. Note that neither of $\ell, d, t$ belongs to a streak. By Observation~\ref{obs:one_value}, for the value in a cell $\ell$, $d$, or $t$ to be equal to one, the cell must neighbour a block in a zero-streak. For each block neighbouring the cells $\ell, d, t$, we use one longest common suffix query to decide whether they are contained in a $0$-streak. If they are, then we can compute the value in that cell and if it equals one, the value in $c$ and all the cells in the bottom row of $B_i$ equal one as well. 

Suppose now that $B_i$ does not belong to a streak. For $D[M,\ell]$ to be equal to one, it must neighbour a block in a $0$-streak. Therefore, there can be only one such cell in $B_i$, the one in the left bottom corner, and we can decide whether the value in it equals to one in constant time similar to above. 
\qed
\end{proof}

\section{}
\label{sec:omitted}
\repeatlemma{basicrecursion}
\begin{proof}
If $i = 0$, then for all $j$, $D[i,j]$ equals the minimum distance between the empty prefix of $P$ and a suffix of $T[1 \dd j]$, which is zero by the definition. If $i > 1$ and $j = 0$, then $D[i,j]$ equals the minimum distance between a non-empty prefix of $P$ and the empty string, which is $\infty$ by the definition.

Assume $i, j \geq 1$. Let us show that $D[i,j] \ge 
\min\{D[i-1,j-1],D[i-1,j], D[i,j-1]\}+ d(P[i], T[j])$ and $D[i,j] \le 
\min\{D[i-1,j-1],D[i-1,j], D[i,j-1]\}+ d(P[i], T[j])$, which implies equality. 
We start by showing the first inequality. Recall that $D[i,j]$ is the smallest $\dtw$ distance between $P[1 \dd i]$ and a suffix of $T[1 \dd j]$. Let this minimum be realised by a suffix $T[j' \dd j]$, where $1 \le j' \le j$ (by definition, $T[j' \dd j]$ is not empty: the distance from $P[1 \dd i]$ to a non-empty suffix is finite, while that to the empty suffix equals $\infty$). Let $\pi$ be a warping path such that its cost equals $\dtw(P[1 \dd i],T[j' \dd j])$. Consider the last edge in $\pi$. If it is from $(i-a,j-b)$ to $(i,j)$, where $a,b \in \{0,1\}$ and $a+b>0$, then 
\begin{align*}
\dtw(P[1 \dd i],&T[j' \dd j]) \\
&\ge d(P[i], T[j]) + \dtw(P[1 \dd i-a],T[j' \dd j-b]) \\
& \ge d(P[i], T[j])  + D[i-a,j-b]\\
& \ge d(P[i], T[j]) + \min\{D[i-1,j-1],D[i-1,j], D[i,j-1]\}
\end{align*}

We now show the second inequality. Let $D[i-a,i-b] = \min\{D[i-1,j-1],D[i-1,j], D[i,j-1]\}$, where $a,b \in \{0,1\}$ and $a+b > 0$. Assume that $D[i-a,j-b]$ is realised on $P[1\dd i-a]$ and $T[j'\dd j-b]$ and a warping path~$\pi$. We can then consider a warping path $\pi' = \pi \cup e$, where $e$ is an edge from $(i-a,j-b)$ to $(i,j)$ for $P[1 \dd i]$ and $T[j'\dd j]$. We have 

\begin{align*}
D[i,j] &\le \dtw(P[1 \dd i], T[j'\dd j]) \le \sum_{(x,y) \in \pi'} d(P[x],T[y])\\  
&= d(P[i],T[j]) + \sum_{(x,y) \in \pi} d(P[x],T[y]) \\
&= d(P[i],T[j]) + \dtw(P[1\dd i-a],T[j'\dd j-b]) \\
& = d(P[i],T[j]) + D[i-a,j-b] \\
&= \min\{D[i-1,j-1],D[i-1,j], D[i,j-1]\}+ d(P[i], T[j])
\end{align*}
\qed\end{proof}

\repeatlemma{nondecreasing}
\begin{proof}
Let us first give an equivalent statement of the lemma: if $(a,b)$ and $(a+1,b)$ are in the same block, then $D[a,b] \le D[a+1,b]$, and if $(a,b)$ and $(a,b+1)$ are in the same block, then $D[a,b] \le D[a,b+1]$. 

We show the lemma by induction on $a+b$. The base of the induction are the cells such that $a = 0$ or $b = 0$, and for them the statement holds by the definition of $D$. Consider now a cell $(a,b)$, where $a,b \ge 1$. Assume that the induction assumption holds for all cells $(x,y)$ such that $x+y < a+b$. By Lemma~\ref{lm:recursion}, we have:
\begin{align*}
&D[a, b] = \min \{ D[a-1, b-1], D[a-1, b], D[a, b-1]\} +d\\
&D[a+1, b] = \min \{ D[a, b-1], D[a, b], D[a+1, b-1]\} + d\\
&D[a, b+1] = \min \{ D[a-1, b], D[a-1, b+1], D[a, b]\} + d\\
\end{align*}
Assume that $(a,b)$ and $(a+1,b)$ are in the same block. 
We have $D[a,b] \leq D[a, b-1]+d$ and trivially $D[a,b] \leq D[a,b] + d$.
By the induction assumption, $D[a,b-1] \leq D[a+1,b-1]$ (the cells $(a,b-1)$ and $(a+1,b-1)$ must belong to the same block).
Therefore, 
\begin{align*}
D[a+1,b] & = \min \{ D[a, b-1], D[a, b], D[a+1, b-1]\} + d \\
& = \min \{ D[a, b-1] + d, D[a, b] + d, D[a+1, b-1] + d\} \\
& \ge \min \{D[a,b], D[a,b], D[a,b-1]+d\} \\
& \ge \min\{D[a,b], D[a,b], D[a,b]\} = D[a,b]. 
\end{align*}
Assume now that $(a,b)$ and $(a,b+1)$ are in the same block.
We have $D[a,b] \leq D[a-1, b]+d$. Furthermore, as $(a-1,b)$ and $(a-1,b+1)$ are in the same block, we have $D[a-1,b] \leq D[a-1,b+1]$ by the induction assumption. Therefore,
\begin{align*}
D[a,b+1] & = \min \{ D[a-1, b], D[a-1, b+1], D[a, b]\} + d\\
& = \min \{ D[a-1, b] + d, D[a-1, b+1] + d, D[a, b] + d\}\\
& \ge \min \{D[a-1,b]+d, D[a-1,b]+d, D[a,b]\}\\ 
& \ge \min\{D[a,b], D[a,b], D[a,b]\} = D[a,b]. 
\end{align*}
This concludes the proof of the lemma.
\qed\end{proof}

\section{}
\label{app:approximate}
\repeattheorem{approx}
\begin{proof}
Any metric $\mu$ can be embedded in $\Oh(\sigma^2)$ time into a well-separated tree metric $\mu_\tau$ of depth $\Oh(\log \sigma)$ with expected distortion $\Oh(\log \sigma)$ (see~\cite{embedding} and~\cite[Theorem 2.4]{treedepth}). Furthermore, the ratio between the smallest distance and the largest distance grows at most polynomially. Formally, for any two letters $a, b$ we have $\mu(a,b) \le \mu_\tau(a,b)$ and $\mathbb{E}(\mu_\tau(a,b)) \le \Oh(\log \sigma) \cdot d(a,b)$. 
Therefore, we have: 
\begin{align}
\label{eq:embedding_lb}
\dtw_{\mu}(X,Y) &\le \dtw_{\mu_\tau}(X,Y)
\end{align}
\begin{align}
\label{eq:embedding_ub}
\mathbb{E}(\dtw_{\mu_\tau}(X,Y)) &\le \Oh(\log \sigma) \cdot \dtw_\mu (X,Y)
\end{align}
Let $\delta = \min_{S-\text{ substr. of }T} \dtw_\mu (P,S)$ and $\delta_\tau = \min_{S-\text{ substr. of }T} \dtw_{\mu_\tau} (P,S)$. Assume that $\delta$ is realised on a substring $X$, and $\delta_\tau$ on a substring $X_\tau$. By Eq.~\ref{eq:embedding_lb}, we then obtain:
$$\delta = \dtw_\mu(P,X) \le \dtw_\mu (P,X_\tau) \le \delta_\tau$$
And Eq.~\ref{eq:embedding_ub} gives the following:
$$\mathbb{E}(\delta_\tau) \le \mathbb{E}(\dtw_{\mu_\tau} (P,X)) \le \Oh(\log \sigma) \cdot \dtw_\mu(P,X) = \Oh(\log \sigma) \cdot \delta$$
We apply the embedding $\log L$ times independently to obtain well-separated tree metrics $\mu_\tau^i$, $i = 1, 2, \ldots, \log L$. From above and by Chernoff bounds, 
$$\min_i \min_{S-\text{ substring of }T} \dtw_{\mu_\tau}^i(P,S)$$
gives an $\Oh(\log \sigma) = \Oh(\log L)$ approximation of $\delta$ with high probability and can be computed in time $\Oh (L^{1-\eps} \cdot mn \log^3 L)$ by Lemma~\ref{lm:approx}, concluding the proof of the theorem. 
\qed\end{proof}

\end{subappendices}

\end{document}